\documentclass{comjnl}

\usepackage{graphicx}
\usepackage{cite}
\usepackage{amsmath}
\usepackage{amssymb}
\usepackage{array}
\usepackage{color, xcolor}
\usepackage{cite}
\usepackage{multirow}
\usepackage[flushleft]{threeparttable}
\usepackage{epstopdf}
\usepackage[T1]{fontenc}

\usepackage{aecompl}

\newcommand{\A}{\ensuremath{\mathbf{A}}}

\begin{document}
\title{Resistance Distances in Simplicial Networks}
\author{Mingzhe~Zhu} 
\author{Wanyue~Xu} 
\author{Zhongzhi~Zhang} \email{zhangzz@fudan.edu.cn}
\author{Haibin~Kan} 
\affiliation{Shanghai Key Laboratory of Intelligent Information
	Processing \& Shanghai Engineering Research Institute of Blockchain, School of Computer Science, Fudan University, Shanghai 200433, China}
\author{Guanrong~Chen} 
\affiliation{Department of Electrical Engineering, City University of Hong Kong, Hong Kong SAR,  China}
\shortauthors{M. Zhu, W. Xu, Z. Zhang, H. Kan and G. Chen}

\keywords{Effective resistance, Kirchhoff index, simplicial network, simplicial complex, higher-order organization, scale-free network}

\begin{abstract}
It is well known that in many real networks, such as brain networks and scientific collaboration networks, there exist higher-order nonpairwise relations among nodes, i.e., interactions between among than two nodes at a time. This simplicial structure can be described by simplicial complexes and has an important effect on topological and dynamical properties of networks involving such group interactions. In this paper, we study analytically resistance distances in iteratively growing networks with higher-order interactions characterized by the simplicial structure that is controlled by a parameter $q$. We derive exact formulas for interesting quantities about resistance distances, including Kirchhoff index, additive degree-Kirchhoff index, multiplicative degree-Kirchhoff index, as well as average resistance distance, which have found applications in various areas elsewhere. We show that the average resistance distance tends to a $q$-dependent constant, indicating the impact of simplicial organization on the structural robustness measured by average resistance distance.
\end{abstract}

\maketitle

\section{Introduction}\label{sec:introduction}

Complex networks have become a powerful tool for studying complex systems, whose nodes represent elements and edges describe their interactions~\cite{Ne03,Ba16}. In most existing work, complex systems are studied by using network models that only capture pairwise relationship among system elements. However, it was  shown~\cite{BeAbScJaKl18,SaCaDaLa18} that in many realistic scenarios, the system structure involves interactions taking place among more than two entities at a time~\cite{BeGlLe16, GrBaMiAl17, LaPeBaLa19}, which are often called higher-order interactions or simplicial interactions. For example, in a scientific collaboration network~\cite{PaPeVa17}, the $q$ authors of a paper form  a $q$-clique, which cannot be described by pairwise interactions, but  is more adequate to be represented by simplicial structure. Other examples involving simplicial interactions include correlations in neuronal spiking activities~\cite{GiPaCuIt15,ReNoScect17}, interactions among proteins~\cite{WuOtBa03}, and so on.

The higher-order interactions in complex systems can be modeled via a collection of simplicial complexes~\cite{CoBi17,PeBa18}.  As generalized network structures, simplicial complexes are not only formed by nodes and links but also by triangles, tetrahedra, and other cliques. They have thus become popular to model complex systems involving interactions among groups~\cite{GiGhBa16, CoBi17,PeBa18, CoBi18, daBiGiDoMe18,LaPeBaLa19,QiYiZh19}. The models generated by simplicial complexes can display simultaneously scale-free behavior, small-world property, and finite spectral dimensions~\cite{BiRa17}. In addition to describing higher-order interactions of complex systems, simplicial structure can also be adopted to study the impact of nonpairwise interactions on collective dynamics of such systems. Recently, it was demonstrated that even three-way interactions can give rise to a host of novel phenomena that are unexpected if only pairwise interactions are considered, for example, Berezinkii-Kosterlitz-Thouless percolation transition~\cite{BiZi18}, abrupt desynchronization~\cite{SkAr19}, and abrupt phase transition of epidemic spreading~\cite{MaGoAr20}.

A simplicial complex is a collection of simplices glued between their faces, where a $q$-dimensional simplex is a $(q+1)$-clique. More recently, inspired by simplicial complexes, a deterministic network model was proposed to describe higher-order interactions, based on an edge operation~\cite{WaYiXuZh19}. Given a graph, the edge operation is defined as follows: for each edge, create a $q$-clique and then connect all $q$ nodes to both ends of the edge. Iteratively using this edge operation to a $(q+2)$-clique generates a model for complex networks with higher-order interactions characterized by $q$. Since the resulting networks consist of simplexes, they are called simplicial networks. They exhibit similar properties as simplicial complexes, include scale-free small-world characteristics and a finite $q$-dependent spectral dimension, highlighting the role of simplicial interactions. Moreover, its normalized Laplacian spectrum can be explicitly determined. It thus serves as an exactly solvable model for simplicial interactions, on which dynamical processes (e.g., random walks) can be studied analytically to shed light on the effect of simplicial interactions.

In this paper, we provide an in-depth study on the properties of resistance distances in simplicial networks, which have found a large variety of applications and thus attracted considerable attention~\cite{SpSr11,DoBu12,YoScLe15,ThYaNa18,DoSibu18,ShZh19,SoHiLi19}. We first formulate recursive expressions for some related matrices, based on which we derive  evolution relations of two-node resistance distances, expressing the resistance distance between two nodes in the current network in terms of those of the node pairs in the previous generation. We then provide explicit expressions for Kirchhoff index, additive degree-Kirchhoff index, and multiplicative degree-Kirchhoff index for simplicial networks. We show that the average resistance distance converges to a $q$-dependent constant.  Thus, when studying complex systems with higher-order organizations, it is necessary to take into account their simplicial interactions.  This work provides rich insights for understanding real systems with simplicial interactions.
$\mathcal{G}_q(t)$ can also be considered as a pure $(q+1)$-dimensional simplicial complex, by looking upon each $(q+1)$-simplex and all its faces as the constituent  simplices.
The properties and high-order interaction of a simplicial complex $\mathcal{K}$ of $N$ nodes can be characterized by a corresponding weighted graph $\mathcal{G'}$~\cite{CoBi16,CoBi17,CaBaCeFa20}, which is obtained from the underlying graph $\mathcal{G}$ of $\mathcal{K}$ by assigning an approximate weight to each edge in $\mathcal{G}$.

\section{Preliminaries}\label{RanWalk}

In this section, we give a brief introduction to some basic concepts related to graphs,  graph Laplacian, and resistance distances.  

\subsection{Graph and Matrix Notation}

Let $\mathcal{G}=(\mathcal{V},\,\mathcal{E})$ be a connected graph with node set $\mathcal{V}$ and edge set~$\mathcal{E} \subset \mathcal{V}\times \mathcal{V}$, for which the number of nodes is $N = |\mathcal{V}|$ and the number of edges is $M=|\mathcal{E}|$. Then, the mean degree of all nodes in $\mathcal{G}$ is $2M /N$.
The $N$ nodes in graph $\mathcal{G}$ will be labeled  by $1,2,3,\ldots, N$.

The  adjacency matrix $A=(a_{ij})_{N \times N}$ of $\mathcal{G}$ captures the adjacency relation among the $N$ nodes, where the entry at row $i$ and column $j$ is defined as follows: $a_{ij}=1$ if the two nodes $i$ and $j$ are connected by an edge in $\mathcal{G}$, and $a_{ij} = 0$ otherwise. For a node $i \in \mathcal{G}$, let $\mathcal{N}(i) = \{x|(x,i)\in \mathcal{E}\}$ denote the set of its neighboring nodes. Then, the degree of node $i$ is $d_i=\sum_{j  \in \mathcal{N}(i)} a_{ij}= \sum_{j=1}^{N} a_{ij} $. Let $D$ denote the diagonal degree matrix of $\mathcal{G}$, where the $i$th diagonal entry is the degree of node $i$, denoted as $d_i$, while all other entries are zeros. The Laplacian matrix of $\mathcal{G}$ is defined as $L = D - A$. 

For a real symmetric square matrix  $X$, not necessary invertible, one can define its $\{1\}-$inverse~\cite{Ti94}. Matrix $M$ is called a $\{1\}-$inverse of $X$, if and only if $XMX = X$.
Let $X^{\dag}$ denote a $\{1\}-$inverse of $X$. The following lemma gives a $\{1\}-$inverse of a block matrix $X$~\cite{Li19}.
\begin{lemma}\label{LemmaBlock1inv}
	For a block matrix
	$X = \left(
	\begin{array}{cc}
	A & B\\
	B^{\top} & C \\
	\end{array}
	\right)
	$,
	where $C$ is nonsingular, if there exists a $\{1\}$-inverse $S^{\dag}$ with $S = A-BC^{-1}B^{\top}$,
	then
	\begin{align}
	X^{\dag} = \left(
	\begin{array}{cc}
	S^{\dag} & -S^{\dag}BC^{-1}\\
	-C^{-1}B^{T}S^{\dag} & C^{-1}B^TS^{\dag}BC^{-1} + C^{-1}
	\end{array}
	\right)\notag
	\end{align}
	is a $\{1\}$-inverse of $X$.
\end{lemma}

\subsection{Resistance Distances and Graph Invariants}

For any graph $\mathcal{G}=(\mathcal{V},\,\mathcal{E})$, if we replace each edge in $\mathcal{E}$ by a unit resistor, we obtain an electrical network~\cite{DoSn84}. For a pair of nodes $i$ and $j$ ($i \neq j$) of $\mathcal{G}$,
the effective resistance $\Omega_{ij}$ between them is defined as the potential difference between $i$ and $j$ when a unit current from $i$ to $j$ is maintained in the corresponding electrical network. In the case $i = j$, $\Omega_{ij}$ is defined to be zero. The  effective resistance $\Omega_{ij}$  is called the resistance distance~\cite{KlRa93} between $i$ and $j$  of  graph  $\mathcal{G}$.

For a graph $\mathcal{G}$, the effective resistance between any node pair  can be represented in terms of  the elements of any \{1\}\textendash inverse of its Laplacian matrix~\cite{Ba99}.
\begin{lemma}\label{efpro1}
	For a  graph $\mathcal{G}$, let $L^{\dagger}_{ij}$ denote the $(i,j)$th entry of any \{1\}\textendash inverse $L^{\dagger}$ of its Laplacian matrix $L$. Then, for any pair of nodes $i,j\in \mathcal{V}$, the effective resistance $\Omega_{ij}$ can be obtained from the elements of $L^{\dagger}$ as
	\begin{equation}
	\Omega_{ij}=L^{\dagger}_{ii}+L^{\dagger}_{jj}-L^{\dagger}_{ij}-L^{\dagger}_{ji}.
	\end{equation}
\end{lemma}

The properties of  resistance distance have been extensively studied, and various sum rules have been established~\cite{Kl02}.
\begin{lemma}\label{Foster}(Foster's First Theorem~\cite{FoRo1949}).
	For a graph $\mathcal{G}=(\mathcal{V},\mathcal{E})$ with $N$ nodes and $M=|\mathcal{E}|$ edges, the sum of resistance distances over all $M$ pairs of adjacent nodes is $N-1$, that
	is,
	\begin{equation}
	\sum_{\substack{i<j\\(i,j)\in\mathcal{E}}}\Omega_{ij}=N-1.
	\end{equation}
\end{lemma}
Theorem~\ref{Foster} was later generalized by Foster himself in~\cite{Fo61},  which is called
Foster's second theorem. In~\cite{ThYaNa18}, further extensions were provided for these two Foster's
theorems.
\begin{lemma}\label{basic}(Sum rule~\cite{Ch10}).
	For any two different  nodes $i$ and $j$  in a graph $\mathcal{G}=(\mathcal{V},\mathcal{E})$,
	\begin{equation}
	d_{i}\Omega_{ij}+\sum_{k\in{\mathcal{N}(i)}}(\Omega_{ik}-\Omega_{jk})=2.
	\end{equation}
\end{lemma}

The resistance distance is an important quantity~\cite{GhBoSa08}, based on which various graph invariants  have been defined and studied. Among these invariants, the Kirchhoff index~\cite{KlRa93} is of vital importance.
For a graph $\mathcal{G}$, its  Kirchhoff index  is defined as
\begin{equation*}
\mathcal{K}(\mathcal{G})=\sum_{i,j=1}^{N}\Omega_{ij}=\sum_{\substack{i\in\mathcal{V}\\j\in\mathcal{V}}}\Omega_{ij}.
\end{equation*}

Kirchhoff index has found many
applications. For example, it can be used as
measures of the overall connectedness of a network~\cite{TiLe10}, the edge
centrality of complex networks~\cite{LiZh18}, as well as the robustness of the
first-order consensus algorithm in noisy networks~\cite{PaBa14, QiZhYiLi19,YiZhPa20}.
The first-order and second-order consensus problems have received considerable
attention from the scientific community~\cite{ShCaHu18, ZhXuYiZh22,YiYaZhZhPa22}.

In recent years,  several modifications of the Kirchhoff index have been proposed, including multiplicative degree-Kirchhoff index~\cite{ChZh07} and additive degree-Kirchhoff index~\cite{GuFeYu12}. For a graph $\mathcal{G}=(\mathcal{V},\mathcal{E})$, its multiplicative degree-Kirchhoff index $R^\ast(\mathcal{G})$ and additive degree-Kirchhoff index $R^+(\mathcal{G})$ are defined as
\begin{equation*}
R^\ast(\mathcal{G})=\sum_{\substack{i\in\mathcal{V}\\j\in\mathcal{V}}}(d_id_j)\Omega_{ij}
\end{equation*}
and
\begin{equation*}
R^+(\mathcal{G})=\sum_{\substack{i\in\mathcal{V}\\j\in\mathcal{V}}}(d_i+d_j)\Omega_{ij},
\end{equation*}
respectively.

It has been shown that the multiplicative degree-Kirchhoff index $R^\ast(\mathcal{G})$ of a graph $\mathcal{G}$ is equal to $4M$ times the Kemeny constant of  the graph~\cite{ChZh07}. The Kemeny constant has been applied to different areas~\cite{Hu14,XuShZhKaZh20}. For example, it can be used as a metric  of  the efficiency of user navigation through the World Wide Web~\cite{LeLo02}.  Also, it was used to measure the efficiency of robotic surveillance in network environments~\cite{PaAgBu15} and to characterize the noise robustness of a class of protocols for formation control~\cite{JaOl19}.

\section{Network Construction and Properties}

The network family   studied here is constructed in an iterative way, controlled by two parameters: $q$ and $t$ with $q\geqslant1$ and $t \geqslant 0$.   Let $\mathcal{G}_q(t)$  be the network after $t$ iterations. Let $\mathcal{K}_q$ ($q\geqslant1$) denote the complete graph with $q$ nodes. When $q=1$, for simplicity, suppose that $\mathcal{K}_1$ is a graph with an isolate node. Then, $\mathcal{G}_q(t)$ is constructed as follows.

\begin{definition}
	For $t=0$ , $\mathcal{G}_q(0)$ is the complete graph $\mathcal{K}_{q+2}$. For $t\geqslant 0$, $\mathcal{G}_q(t+1)$ is obtained from $\mathcal{G}_q(t)$  by performing the following operation (see Fig.~\ref{build}): for every
	existing edge of $\mathcal{G}_q(t)$, introduce a copy of the
	complete graph $\mathcal{K}_q$ and connect all its nodes to both
	end nodes of the edge. 
\end{definition}
Figures~\ref{netA} and~\ref{netB} illustrate the networks  for  two particular cases of $q =1$ and $q =2$, respectively. 

\begin{figure}
	\centering
	\includegraphics[width=1.0\linewidth]{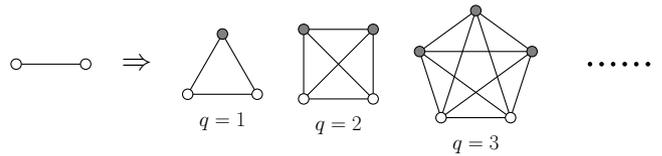}
	\caption{Network construction method. The next-iteration network is obtained from the current network by performing the operation on the right-hand side of the arrow  for each existing edge.}
	\label{build}
\end{figure}

\begin{figure}
	\centering
	\includegraphics[width=0.30\textwidth]{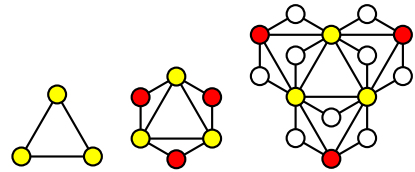}
	\caption{The networks of the first three iterations for $q =1$.} 
	\label{netA}
\end{figure}

\begin{figure}
	\centering
	\includegraphics[width=0.35\textwidth]{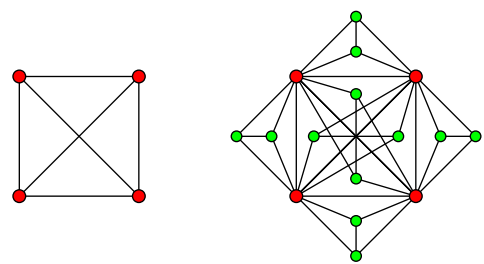}
	\caption{The networks of the first two iterations for $q =2$.} 
	\label{netB}
\end{figure}


For network  $\mathcal{G}_q(t)$, let  $\mathcal{V}_{g}$ denote its  node set, and let $\mathcal{E}_{g}$ denote its edge set.   Let $N_t= |\mathcal{V}_{g}|$ and $M_t=|\mathcal{E}_{g}|$ denote, respectively, the number of nodes and the number of edges in graph $\mathcal{G}_q(t)$.  It is easy to verify that? for all $t\geqslant0$,
\begin{gather}
M_t=\left[\frac{(q+1)(q+2)}{2}\right]^{t+1}\label{M},\\
N_t\!=\!\frac{2}{q+3}\left[\frac{(q+1)(q+2)}{2}\right]^{t+1}\!+\!\frac{2(q+2)}{q+3}.\label{N}
\end{gather}
Then, the average degree of  graph $G_q(t)$ is $2M_q(t)/N_q(t)$, which converges to $q + 3$ when $t$ is sufficiently  large. Therefore, the graph family $G_q(t)$ is sparse.

Let $\mathcal{W}_{t+1}=\mathcal{V}_{t+1}\backslash\mathcal{V}g$ represent the set of new nodes generated at iteration $g+1$, and let $W_{t+1}=|\mathcal{W}_{t+1}|$ stand for the number of these newly generated nodes. Then,
\begin{equation}\label{W}
W_{t+1}=q\left[\frac{(q+1)(q+2)}{2}\right]^{t+1}.
\end{equation}
Let $d^{(t)}_v$ be the degree of a node $v$ in graph $G_q(t)$,  which was generated at iteration $t_v$. Then,
\begin{gather}\label{d_v}
d^{(t)}_v=(q+1)^{t-t_v+1}.
\end{gather}
By construction, in graph $G_q(t)$ the degree of all simultaneously emerging nodes is the same. Thus, the number of nodes with  degree $(q+1)^{t-t_v+1}$ is $q+2$ and $q\left[\frac{(q+1)(q+2)}{2}\right]^{t_v}$ for $t_v=0$ and $t_v>0$, respectively.

By construction, the resulting networks consist of cliques $K_{q+2}$ or smaller cliques and are thus called simplicial networks, characterized by parameter $q$.  They display some remarkable properties as observed in most real networks~\cite{Ne03}. They are scale-free, since their node degrees obey a power-law distribution $P(d)\sim d^{-\gamma_q}$ with $\gamma_q=2 +\frac{\ln (q+2)}{\ln (q+1)}-\frac{\ln 2}{\ln (q+1)}$~\cite{WaYiXuZh19}. They are  small-world, since their diameters grow logarithmically with the number of nodes   and their mean clustering coefficients approach a large constant $\frac{q^2+3q+3 }{q^2+3q+5}$~\cite{WaYiXuZh19}. Moreover, they have a finite spectral dimension $ \frac{2[\ln(q^2+3q+3)-\ln 2]}{\ln (q+1)}$. Thus, many features of  simplicial networks are related to higher-order interactions encoded in  parameter $q$.   It is expected that other properties are also dependent on group  interactions, including the resistance distance to be studied in the following. 

\section{ Relations among Various Matrices}

Before determining  resistance distances, Kirchhoff index and its invariants in simplicial networks,  we first provide some relations among matrices related to  simplicial networks. These relations are  very useful for deriving the properties of  resistance distances, as well as those relevant quantities derived from  resistance distances.

Let $A_t$  denote the adjacency matrix of graph $\mathcal{G}_q(t)$. Its element $A_t(i,j)$ at row $i$ and column $j$ of $A_t$ is:  $A_t(i,j)=1$ if nodes $i$ and $j$ are connected by an edge in  $\mathcal{G}_q(t)$, $A_t (i,j)=0$ otherwise. Let $D_t$  denote the diagonal degree matrix of matrix $\mathcal{G}_q(t)$, with the $i$th diagonal element being the degree $d_i^{(t)}$ of node $i$. Let $L_t$ denote the Laplacian matrix of $\mathcal{G}_q(t)$. Then, $L_t=D_t-A_t$. Next, the recursion relations are derived for the these matrices $A_t$, $D_t$, and $L_t$.

For graph $\mathcal{G}_{q}(t+1)$, let $\alpha$ be the set of old nodes that are already present in $\mathcal{G}_q(t)$, and $\beta$ the set of new nodes generated at iteration $t+1$,  namely, those nodes  in $\mathcal{W}_{t+1}$. Then, write $A_{t+1}$  in block form as
\begin{equation}
	A_{t+1}=
	\left
	[\begin{array}{cc}
	A_{t+1}^{\alpha,\alpha} & A_{t+1}^{\alpha,\beta}\\
	A_{t+1}^{\beta,\alpha} & A_{t+1}^{\beta,\beta}
	\end{array}
	\right],
\end{equation}
where $A^{\alpha,\alpha}_{t+1}=A_t$, $A^{\alpha,\beta}_{t+1}=(A^{\beta,\alpha}_{t+1})^\top$, and $A^{\beta,\beta}_{t+1}$ is a diagonal block matrix of order/dimension $W_{t+1}/ q$, taking the form  $A^{\beta,\beta}_{t+1}={\rm {diag}}(B_q,B_q,\ldots,B_q)$
with $B_q$ being  the adjacency matrix of the complete graph $\mathcal{K}_q$ for $q \geq 1$.

In what follows, let $I$ denote the identity matrix of approximate dimension. Then, the diagonal matrix $D_{t+1}$ is given by
\begin{flalign}
    &D_{t+1}\!=\!
	\begin{bmatrix}
		D^{\alpha,\alpha}_{t+1}\! & \!0\\
		0\! &\! (q\!+\!1)I
	\end{bmatrix}
	\!=\!
	\begin{bmatrix}
	(q\!+\!1)D_{t}\! &\! 0\\
	0 \!& \!(q\!+\!1)I
	\end{bmatrix},
\end{flalign}
which is obtained based on the following facts:  when the network evolves from iteration $t$ to iteration $t+1$, the degree of each node in set $\alpha$ increases by a factor of $q+1$ as shown in~\eqref{d_v}, while the degree of all nodes in set $\beta$ is equal to $q+1$. Therefore,  the Laplacian matrix $L_{t+1}$ satisfies
\begin{small}
	\begin{flalign}\label{lg+1}
	L_{t+1}=&D_{t+1}-A_{t+1}\notag\\
	=&
	\begin{bmatrix}
	(q\!+\!1)D_{t}\!-\!A_t\! &\! -A^{\alpha,\beta}_{t+1}\\
	-A^{\beta,\alpha}_{t+1} \!& \!(q\!+\!1)I\!-\!A^{\beta,\beta}_{t+1}
	\end{bmatrix}.
	\end{flalign}
\end{small}

In addition to the above-derived recursion relations for adjacency matrix, degree diagonal matrix, and Laplacian matrix, relations for other relevant matrices can also be established.

\begin{lemma}\label{PFProA}
For  graph $\mathcal{G}_q(t+1)$, $t\geqslant 0$,
	 \begin{equation}
	 	A^{\alpha,\beta}_{t+1}A^{\beta,\alpha}_{t+1}=q(D_t+A_t).
	 \end{equation}
\end{lemma}
The proof of Lemma~\ref{PFProA} is provided in~\ref{AppA}.

\begin{lemma}\label{beauty}
For  graph $\mathcal{G}_q(t+1)$, $t\geqslant 0$,
	\begin{equation}
		A^{\alpha,\beta}_{t+1}\left((q+1)I-A^{\beta,\beta}_{t+1}\right)^{-1}=\frac{1}{2}A^{\alpha,\beta}_{t+1}.
	\end{equation}
\end{lemma}
The proof of Lemma~\ref{beauty} is provided in~\ref{AppB}.

\section{ Relations among Effective Resistances}

In this section, we study the relations governing  resistance distances.  For graph $\mathcal{G}_q(t+1)$, we first establish the evolution rule of  resistance distance between any pair of old nodes in $\mathcal{G}_q(t)$.  Then, we demonstrate that the resistance distance between two arbitrary  nodes in $\mathcal{G}_q(t+1)$ can be exactly determined or expressed  in terms of resistance distances of those old node pairs in $\mathcal{G}_q(t)$.

Write $\Omega^{(t)}_{ij}$ to represent   the resistance distance  between  nodes $i$ and $j$ in graph $\mathcal{G}_q(t)$, and  write $L^{\dagger}_t$ to denote a \{1\}\textendash inverse of Laplacian matrix $L_t$ for $\mathcal{G}_q(t)$.
\begin{lemma}\label{lemma7}
Let $i,j\in \mathcal{V}_t$ be a pair of old nodes in $\mathcal{G}_{t+1} (t\geqslant0)$. Then, $\Omega^{(t)}_{ij}$ satisfies the following recursive relation:
	\begin{equation}
		\Omega^{(t+1)}_{ij}=\frac{2}{q+2}\Omega^{(t)}_{ij}.
	\end{equation}
\end{lemma}
\begin{proof}
Any \{1\}\textendash inverse $L^{\dagger}_{t+1}$ of matrix $L_{t+1}$ can be written as
	\begin{equation}
	L^{\dagger}_{t+1}=
	\begin{bmatrix}
	L^{\dagger}_{\alpha,\alpha} & L^{\dagger}_{\alpha,\beta}\\
	L^{\dagger}_{\beta,\alpha} & L^{\dagger}_{\beta,\beta}
	\end{bmatrix}.
	\end{equation}
By~\eqref{lg+1}  and Lemmas~\ref{efpro1}, ~\ref{PFProA} and~\ref{beauty}, one obtains
	\begin{small}
		\begin{align}\label{PFoosubmat}
L^{\dagger}_{\alpha,\alpha}=&\left(\!(q\!+\!1)D_t\!-\!A_t\!-\!A^{\alpha,\beta}_{t+1}\!\left((q\!+\!1)I\!-\!A^{\beta,\beta}_{t+1}\right)^{\!-\!1}\!\!\!A^{\beta,\alpha}_{t+1}\right)^{\dagger}\notag\\
		=&\left((q+1)D_t-A_t-\frac{q}{2}\left(D_t+A_t\right)\right)^{\dagger}\notag\\
		=&\left(\left(\frac{q}{2}+1\right)(D_t-A_t)\right)^{\dagger}\notag\\
		=&\frac{2}{q+2}L^{\dagger}_{t}.
		\end{align}
	\end{small}
\end{proof}
By Lemma~$\ref{efpro1}$ and~\eqref{PFoosubmat}, for two nodes~$i,j \in \mathcal{V}_t$, one has
\begin{align}
&\Omega_{ij}(t+1) \nonumber\\
=& L^{\dag}_{\alpha,\alpha}(i,i) + L^{\dag}_{\alpha,\alpha}(j,j)
- L^{\dag}_{\alpha,\alpha}(i,j) - L^{\dag}_{\alpha,\alpha}(j,i)\nonumber\\
=& \frac{2}{q+2} \left( L^{\dag}_{t}(i,i) + L^{\dag}_{t}(j,j)
- L^{\dag}_{t}(i,j) - L^{\dag}_{t}(j,i) \right) \nonumber\\
=& \frac{2}{q+2} \Omega_{ij}^{(t)}\,,
\end{align}
 as required.  


In addition to the pairs of old nodes, the effective resistance between any other pairs of  nodes in $\mathcal{G}_q(t+1)$ can be explicitly determined or be represented in terms of those for old nodes in $\mathcal{G}_q(t)$.
 To this end, introduce some additional quantities. For any two subsets $X$ and $Y$ of set $\mathcal{V}_t$ for  nodes in graph $\mathcal{G}_q(t)$, define
\begin{equation}
	\Omega^{(t)}_{X,Y}=\sum_{i\in X,j\in Y}\Omega^{(t)}_{ij}.
\end{equation}
For a new node $i\in \mathcal{W}_{t+1}$ in $\mathcal{G}_q(t+1)$, let $\Delta_{i}=\{m,n\}$ be the set of two old neighbors of $i$. Define
\begin{equation}
	\Omega^{(t+1)}_{\Delta_i}=\Omega^{(t+1)}_{mn}.
\end{equation}

\begin{lemma}\label{lemma8}
For $t\geqslant0$, $i,j\in \mathcal{W}_{t+1}$ that are adjacent to each other, one has
	\begin{equation}\label{Omegaijw}
		\Omega^{(t+1)}_{ij}=\frac{2}{q+2}.
	\end{equation}
\end{lemma}
\begin{proof}
by Lemma~\ref{basic}, for $i\in \mathcal{W}_{t+1}$ and its neighboring node $j\in \mathcal{W}_{t+1}$, one obtains
	\begin{equation}\label{Omegaijx}
		(q+1)\Omega^{(t+1)}_{ij}+\sum_{k\in{\mathcal{N}(i)}}\left(\Omega^{(t+1)}_{ik}-\Omega^{(t+1)}_{jk}\right)=2.
	\end{equation}
By symmetry,   for any node $k \in \mathcal{N}(i)$ except $j$, one has
	\begin{equation}
		\Omega^{(t+1)}_{ik}=\Omega^{(t+1)}_{jk},
	\end{equation}
which leads to
	\begin{equation}\label{Omegaijy}
	\sum_{k\in{\mathcal{N}(i)}}\left(\Omega^{(t+1)}_{ik}-\Omega^{(t+1)}_{jk}\right)=\Omega^{(t+1)}_{ij}.
	\end{equation}
With~\eqref{Omegaijx} and~\eqref{Omegaijy}, one obtains~\eqref{Omegaijw}.
\end{proof}

\begin{lemma}\label{lemma9}
For a node $i\in \mathcal{W}_{t+1}$ with $t\geqslant0$, one has
	\begin{equation}
	\Omega^{(t+1)}_{i,\Delta_{i}}=\frac{3}{q+2}+\frac{1}{2}\Omega^{(t+1)}_{\Delta_{i}}.
	\end{equation}
\end{lemma}
\begin{proof}
by Lemma~\ref{basic}, for $i\in \mathcal{W}_{t+1}$ and its two old neighbors $m$ and $n$ belonging to  $\mathcal{V}_{t}$ and forming set $\Delta_{i}=\{m,n\}$, one has
	\begin{equation}\label{eq91}
	(q+1)\Omega^{(t+1)}_{im}+\sum_{k\in{\mathcal{N}(i)}}\left(\Omega^{(t+1)}_{ik}-\Omega^{(t+1)}_{mk}\right)=2
	\end{equation}
	and
    \begin{equation}\label{eq92}
    (q+1)\Omega^{(t+1)}_{in}+\sum_{k\in{\mathcal{N}(i)}}\left(\Omega^{(t+1)}_{ik}-\Omega^{(t+1)}_{nk}\right)=2.
    \end{equation}
 By symmetry, for any node $k\in\mathcal{N}(i)$ except $m$ and $n$,
$\Omega^{(t+1)}_{mk}=\Omega^{(t+1)}_{mi}$ holds, which  implies that
 \begin{equation}\label{eq93}
    	\sum_{k\in{\mathcal{N}(i)}}\Omega^{(t+1)}_{mk}=(q-1)\Omega^{(t+1)}_{im}+\Omega^{(t+1)}_{\Delta_{i}}.
    \end{equation}
On the other hand, using  Lemma~\ref{lemma8}, one obtains
    \begin{equation}\label{eq94}
    	\sum_{k\in{\mathcal{N}(i)}}\Omega^{(t+1)}_{ik}=\frac{2(q-1)}{q+2}+\Omega^{(t+1)}_{i,\Delta_{i}}.
    \end{equation}
With~\eqref{eq91},~\eqref{eq93},  and~\eqref{eq94}, one has
    \begin{equation}
    	 2\Omega^{(t+1)}_{im}+\Omega^{(t+1)}_{i,\Delta_{i}}-\Omega^{(t+1)}_{\Delta_{i}}=\frac{6}{q+2}.
    \end{equation}
Analogously,
    \begin{equation}
    2\Omega^{(t+1)}_{in}+\Omega^{(t+1)}_{i,\Delta_{i}}-\Omega^{(t+1)}_{\Delta_{i}}=\frac{6}{q+2}.
    \end{equation}
The above two equations show that  $\Omega^{(t+1)}_{im}=\Omega^{(t+1)}_{in}$, which is easily understood according to the network construction. 
 Summing these two equations gives
    \begin{equation}
    4\Omega^{(t+1)}_{i,\Delta_{i}}-2\Omega^{(t+1)}_{\Delta_{i}}=\frac{12}{q+2}.
    \end{equation}
    That is
    \begin{equation}
    \Omega^{(t+1)}_{i,\Delta_{i}}=\frac{3}{q+2}+\frac{1}{2}\Omega^{(t+1)}_{\Delta_{i}},
    \end{equation}
    as required.
\end{proof}

\begin{lemma}\label{lemma10}
For $t\geqslant0$ and two nodes $i$ and $j$, with $i\in\mathcal{W}_{t+1}$ and $j\in\mathcal{V}_t$, one has
	\begin{equation}\label{lemma10eq}
\Omega^{(t+1)}_{ij}=\frac{1}{2}\left(\frac{3}{q+2}-\frac{1}{2}\Omega^{(t+1)}_{\Delta_{i}}+\Omega^{(t+1)}_{\Delta_{i},j}\right).
	\end{equation}
\end{lemma}
\begin{proof}
For the pair of nodes  $i\in\mathcal{W}_{t+1}$ and $j\in\mathcal{V}_t$, by Lemma~\ref{basic}, one obtains
	\begin{equation}\label{eq101}
		(q+1)\Omega^{(t+1)}_{ij}+\sum_{k\in{\mathcal{N}(i)}}\left(\Omega^{(t+1)}_{ik}-\Omega^{(t+1)}_{jk}\right)=2.
	\end{equation}
Considering the symmetry, for any node $k\in\mathcal{N}(i)$ except $m$ and $n$, one has
$\Omega^{(t+1)}_{jk}=\Omega^{(t+1)}_{ji}$, which yields
	\begin{equation}\label{eq102}
	\sum_{k\in{\mathcal{N}(i)}}\Omega^{(t+1)}_{jk}=(q-1)\Omega^{(t+1)}_{ji}+\Omega^{(t+1)}_{j,\Delta_{i}},
	\end{equation}
Moreover, according to Lemma~\ref{lemma8}, one obtains
	\begin{equation}\label{eq103}
	\sum_{k\in{\mathcal{N}(i)}}\Omega^{(t+1)}_{ik}=\frac{2(q-1)}{q+2}+\Omega^{(t+1)}_{i,\Delta_{i}}.
	\end{equation}
Combining~\eqref{eq101},~\eqref{eq102} and~\eqref{eq103} yields
	\begin{equation}
	2\Omega^{(t+1)}_{ij}+\Omega^{(t+1)}_{i,\Delta_{i}}-\Omega^{(t+1)}_{j,\Delta_{i}}=\frac{6}{q+2},
	\end{equation}
which, together with Lemma~\ref{lemma9}, gives
	\begin{equation}
		2\Omega^{(t+1)}_{ij}+\frac{1}{2}\Omega^{(t+1)}_{\Delta_{i}}-\Omega^{(t+1)}_{j,\Delta_{i}}=\frac{3}{q+2},
	\end{equation}
a formula equivalent to~\eqref{lemma10eq}.
\end{proof}

\begin{lemma}\label{lemma11}
For $t\geqslant0$, a pair of nonadjacent  nodes $i$ and $ j$, both in $\mathcal{W}_{t+1}$, satisfy
	\begin{equation} \Omega^{(t+1)}_{ij}=\frac{3}{q+2}-\frac{1}{4}\left(\Omega^{(t+1)}_{\Delta_{i}}+\Omega^{(t+1)}_{\Delta_{j}}\right)+\frac{1}{4}\Omega^{(t+1)}_{\Delta_{i},\Delta_{j}}.
	\end{equation}
\end{lemma}
\begin{proof}
	For two nonadjacent nodes $i$ and $j$ in $\mathcal{W}_{t+1}$, by Lemma~\ref{basic}, one has
	\begin{equation}\label{eq111}
		(q+1)\Omega^{(t+1)}_{ij}+\sum_{k\in{\mathcal{N}(i)}}\left(\Omega^{(t+1)}_{ik}-\Omega^{(t+1)}_{jk}\right)=2.
	\end{equation}
Following the same process as in the proof of Lemma~\ref{lemma10}, one obtains
	\begin{flalign}
	\sum_{k\in{\mathcal{N}(i)}}\Omega^{(t+1)}_{ik}&=\frac{2(q-1)}{q+2}+\Omega^{(t+1)}_{i,\Delta_{i}},\label{eq112}\\
	\sum_{k\in{\mathcal{N}(i)}}\Omega^{(t+1)}_{jk}&=(q-1)\Omega^{(t+1)}_{ji}+\Omega^{(t+1)}_{j,\Delta_{i}}.\label{eq113}
	\end{flalign}
Combining~\eqref{eq111},~\eqref{eq112} and~\eqref{eq113} yields
	\begin{equation}
	2\Omega^{(t+1)}_{ij}+\Omega^{(t+1)}_{i,\Delta_{i}}-\Omega^{(t+1)}_{j,\Delta_{i}}=\frac{6}{q+2}\,,
	\end{equation}
which, together with Lemmas~\ref{lemma9} and~\ref{lemma10}, leads to
	\begin{small}
		\begin{align}
	&\Omega^{(t+1)}_{ij}=\frac{1}{2}\left(\frac{6}{q+2}-\Omega^{(t+1)}_{i,\Delta_{i}}+\Omega^{(t+1)}_{j,\Delta_{i}}\right)\notag\\
		=&\frac{1}{2}\left(\frac{3}{q\!+\!2}\!-\!\frac{1}{2}\Omega^{(t\!+\!1)}_{\Delta_{i}}\!\!+\!\!\sum_{k\in\Delta_{i}}\!\frac{1}{2}\left(\frac{3}{q\!+\!2}\!-\!\frac{1}{2}\Omega^{(t\!+\!1)}_{\Delta_{j}}\!\!+\!\Omega^{(t\!+\!1)}_{\Delta_{j},k}\right)\right)\notag\\
		=&\frac{1}{2}\left(\frac{3}{q\!+\!2}\!-\!\frac{1}{2}\Omega^{(t\!+\!1)}_{\Delta_{i}}\!+\!\left(\frac{3}{q\!+\!2}\!-\!\frac{1}{2}\Omega^{(t\!+\!1)}_{\Delta_{j}}\!+\!\frac{1}{2}\Omega^{(t\!+\!1)}_{\Delta_{i},\Delta_{j}}\right)\right)\notag\\
		=&\frac{3}{q+2}-\frac{1}{4}\left(\Omega^{(t+1)}_{\Delta_{i}}+\Omega^{(t+1)}_{\Delta_{j}}\right)+\frac{1}{4}\Omega^{(t+1)}_{\Delta_{i},\Delta_{j}}.
		\end{align}
	\end{small}
Thus, the result follows.
\end{proof}

\section{Exact Solutions to Various Kirchhoff Indices}

In this section, we determine the multiplicative degree-Kirchhoff index, additive degree-Kirchhoff index, and Kirchhoff index for graph $\mathcal{G}_q(t)$. To do so, we define three more quantities about resistance distances related to
 graph $\mathcal{G}_q(t)$. For two subsets $X $ and $Y $ of set $\mathcal{V}_t$ for nodes in $\mathcal{G}_q(t)$, define
\begin{align}
&R_{X,Y}(t)=\sum_{i\in X,j\in Y}\Omega^{(t)}_{ij},\\
&R^\ast_{X,Y}(t)=\sum_{i\in X,j\in Y}(d_id_j)\Omega^{(t)}_{ij},\\
&R^+_{X,Y}(t)=\sum_{i\in X,j\in Y}(d_i+d_j)\Omega^{(t)}_{ij}.
\end{align}
Then, $R_{\mathcal{V}_t,\mathcal{V}_t}(t)$, $R^\ast_{\mathcal{V}_t,\mathcal{V}_t}(t)$ and $R^+_{\mathcal{V}_t,\mathcal{V}_t}(t)$ are, respectively, the Kirchhoff index, multiplicative degree-Kirchhoff index, and additive degree-Kirchhoff index of graph $\mathcal{G}_q(t)$.

For $t=0$, it is easy to derive $R_{\mathcal{V}_{0},\mathcal{V}_{0}}(0)=2(q+1)$, $R^{\ast}_{\mathcal{V}_0,\mathcal{V}_0}(0)=2(q+1)^3$, and $R^{+}_{\mathcal{V}_{0},\mathcal{V}_{0}}(0)=4(q+1)^2$.
To obtain explicit formulas for  $R_{\mathcal{V}_t,\mathcal{V}_t}(t)$, $R^\ast_{\mathcal{V}_t,\mathcal{V}_t}(t)$, and $R^+_{\mathcal{V}_t,\mathcal{V}_t}(t)$ for all $t\geq 0$, some intermediary results are needed.
\begin{lemma}\label{lemma12}
For graph $\mathcal{G}_q(t+1)$ with $t\geqslant0$,
	\begin{equation}
		\sum_{i\in \mathcal{W}_{t+1}}\Omega^{(t+1)}_{\Delta_{i}}=\frac{2q(N_t-1)}{q+2}.
	\end{equation}
\end{lemma}
\begin{proof}
Note that every edge in $\mathcal{G}_q(t)$ creates exactly $q$ new nodes of $\mathcal{G}_q(t+1)$. Summing $\Omega^{(t+1)}_{\Delta_{i}}$ over $\Delta_{i}$ of all  nodes $i\in  \mathcal{W}_{t+1}$ is equivalent to summing $\Omega^{(t+1)}_{xy}$ for $q$ times over all edges $(x,y)$ belonging to $\mathcal{E}_t$. Then, by Lemma~\ref{Foster}, one obtains
	\begin{equation}
	\begin{split}
		\sum_{i\in \mathcal{W}_{t+1}}\Omega^{(t+1)}_{\Delta_{i}}&=q\sum_{(x,y)\in \mathcal{E}_t}\Omega^{(t+1)}_{xy}\\
		&=q\sum_{(x,y)\in \mathcal{E}_t}\frac{2}{q+2}\Omega^{(t)}_{xy}\\
		&=\frac{2q(N_t-1)}{q+2},
	\end{split}
	\end{equation}
This completes the proof.
\end{proof}
\begin{lemma}\label{lemma13}
For $t\geqslant0$ and a set $Y\subseteq\mathcal{V}_{t}$, one has
	\begin{equation}
		\sum_{i\in \mathcal{W}_{t+1}}\Omega^{(t+1)}_{\Delta_{i},Y}=\sum_{x\in \mathcal{V}_{t}}qd_x^{(t)}\Omega^{(t+1)}_{x,Y}.
	\end{equation}
\end{lemma}
\begin{proof}
For an arbitrary node $x\in\mathcal{V}_t$, there are $d^{(t+1)}_x-d^{(t)}_x=qd^{(t)}_x$ new nodes in $\mathcal{W}_{t+1}$ that are neighbors of $x$. Thus, $\Omega^{(t+1)}_{x,Y}$ is summed $qd^{(t)}_x$  times for each node $x\in \mathcal{V}_{t}$.
\end{proof}

Now, we are in position to prove the main results.
\begin{theorem}\label{lemma14}
For graph $\mathcal{G}_q(t)$ with  $t\geqslant0$, its multiplicative degree-Kirchhoff  index is
	 \begin{small}
	 	\begin{equation}\label{mul}
	 	\begin{aligned}
	 	&R^{\ast}_{\mathcal{V}_t,\mathcal{V}_t}(t)\\
	 	=&-(q+4)(q+1)^2\left(\frac{(q+2)(q+1)^2}{2}\right)^t\\
	 	&+\frac{2(q+2)(q+1)^2}{q+3}\left(\frac{(q+1)(q+2)}{2}\right)^t\\
	 	&+\frac{(q+2)(3q+7)(q+1)^2}{q+3}\left(\frac{(q+1)(q+2)}{2}\right)^{2t}.
	 	\end{aligned}
	 	\end{equation}
	 \end{small}
\end{theorem}
The proof of Theorem~\ref{lemma14} is provided in~\ref{AppC}.

\begin{theorem}\label{lemma15}
For graph $\mathcal{G}_q(t)$ with  $t\geqslant0$,	its additive degree-Kirchhoff index is
	\begin{small}
		 \begin{align}\label{eqxxx}
		&R^{+}_{\mathcal{V}_t,\mathcal{V}_t}(t)\notag\\
		=&\frac{4(q+2)(q+1)^2}{(q+3)^2}\left(\frac{(q+1)(q+2)}{2}\right)^t\notag\\	&+\frac{2(3q\!+\!7)(q\!+\!1)^2(q^3\!+\!8q^2\!+\!22q\!+\!20)}{(q\!+\!3)^2(q^2\!+\!5q\!+\!8)}\left(\frac{(q\!+\!1)(q\!+\!2)}{2}\right)^{2t}\notag\\
		&-\frac{2(q+4)(q+1)^2}{q+3}\left(\frac{(q+2)(q+1)^2}{2}\right)^t\notag\\
		&+\frac{2(q+1)^2(q^2+9q+20)}{(q+3)(q^2+5q+8)}\left(q+1\right)^t		+\frac{2(q+1)^2}{(q+3)^2}.
		\end{align}
	\end{small}
\end{theorem}
The proof of Theorem~\ref{lemma15} is provided in~\ref{AppD}.

\begin{theorem}\label{lemma16}
For graph $\mathcal{G}_q(t)$ with  $t\geqslant0$, its Kirchhoff index is
	\begin{small}
		\begin{align}\label{eq160}
		&R_{\mathcal{V}_t,\mathcal{V}_t}(t)\notag\\
		=&\frac{2(q+2)(q^3+8q^2+15q+8)}{(q+3)^2(q^2+5q+8)}\left(\frac{(q+1)(q+2)}{2}\right)^t\notag\\
		&+\frac{(q+2)(q+4)(3q+7)(q+1)^2}{(q+3)^2(q^2+5q+8)}\left(\frac{(q+1)(q+2)}{2}\right)^{2t}\notag\\
		&-\frac{(q+4)(q+1)^2}{(q+3)^2}\left(\frac{(q+2)(q+1)^2}{2}\right)^{t}\notag\\
		&+\frac{2(q+1)^2(q^2+9q+20)}{(q+3)^2(q^2+5q+8)}\left(q+1\right)^t\notag\\
		&+\frac{4(q+1)(q+4)^2}{(q+3)^2(q^2+5q+8)}\left(\frac{2}{q+2}\right)^t		-\frac{2(q+1)}{(q+3)^2}.
		\end{align}
	\end{small}
\end{theorem}
The proof of Theorem~\ref{lemma16} is provided in~\ref{AppE}.

Using the obtained Kirchhoff index, the average resistance distance $\langle\bar{\Omega}(t) \rangle$ for graph $\mathcal{G}_q(t)$ can be easily determined, as $\langle\bar{\Omega}(t) \rangle=\frac{R_{\mathcal{V}_t,\mathcal{V}_t}(t)}{N_t(N_t-1)}$.
\begin{theorem}\label{th1}
For graph $\mathcal{G}_q(t)$ with  $t\geqslant0$, its mean resistance distance $\langle\bar{\Omega}(t) \rangle$ is
	 \begin{small}
	 	\begin{align}\label{eq170}
	 	&\langle\bar{\Omega}(t) \rangle=\notag\\
 &\frac{(q\!+\!3)^2}{(q\!+\!1)^2(q\!+\!2)^2\!\left(\left(\frac{(q\!+\!1)(q\!+\!2)}{2}\right)^t\!\!\!+\!\frac{2}{q\!+\!1}\right)\!\!\left(\left(\frac{(q\!+\!1)(q\!+\!2)}{2}\right)^t\!\!\!+\!\frac{1}{q\!+\!2}\right)}\notag\\
	 	&\bigg(\frac{2(q+2)(q^3+8q^2+15q+8)}{(q+3)^2(q^2+5q+8)}\left(\frac{(q+1)(q+2)}{2}\right)^t\notag\\
	 	&+\frac{(q+2)(q+4)(3q+7)(q+1)^2}{(q+3)^2(q^2+5q+8)}\left(\frac{(q+1)(q+2)}{2}\right)^{2t}\notag\\
	 	&-\frac{(q+4)(q+1)^2}{(q+3)^2}\left(\frac{(q+2)(q+1)^2}{2}\right)^{t}\notag\\
	 	&+\frac{2(q+1)^2(q^2+9q+20)}{(q+3)^2(q^2+5q+8)}\left(q+1\right)^t\notag\\
	 	&+\frac{4(q+1)(q+4)^2}{(q+3)^2(q^2+5q+8)}\left(\frac{2}{q+2}\right)^t	 	-\frac{2(q+1)}{(q+3)^2}\bigg).
	 	\end{align}
	 \end{small}
\end{theorem}
In the limit of large $t$ ($t\rightarrow\infty$),~\eqref{eq170} converges to a  small $q$-dependent constant
     \begin{equation}\label{eq170a}
   \lim_{t\rightarrow\infty} \langle\bar{\Omega}_q(t) \rangle= \langle\bar{\Omega}_q\rangle=\frac{(q+4)(3q+7)}{(q+2)(q^2+5q+8)},
     \end{equation}
 demonstrating the impact of network geometry and topology.

In Figure~\ref{EffetRes}, we plot the average resistance distance $\langle\bar{\Omega}_q(t) \rangle$ as a function of $q$ and $t$ according to~\eqref{eq170}. Figure~\ref{EffetRes} shows that for all $q \geq  1$ and $t \geq 0$, the average resistance distance $\langle\bar{\Omega}_q(t) \rangle$  is a little greater than zero. Particularly, for massive graphs ($t \to \infty$), the average resistance distance $\langle\bar{\Omega}_q(t) \rangle$ expressed in~\eqref{eq170} tends a small constant $\langle\bar{\Omega}_q\rangle$ given by~\eqref{eq170a}, which
decreases with increasing parameter $q$.

\begin{figure}
\centering
\includegraphics[width=0.9\linewidth,trim=0 0 0 0]{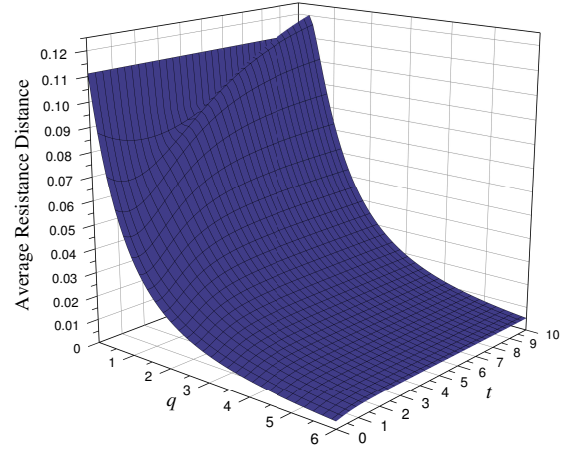}
\caption{ The average resistance distance $\langle\bar{\Omega} _q(t) \rangle$  of  network $\mathcal{G}_q(t)$ for various $q$ and $t$.}\label{EffetRes}
\end{figure}


\section{Conclusion}

Empirical studies have shown that in many realistic networks, such as biological and social networks, their structures involve non-pairwise interactions, that is simultaneous, higher-order interactions among a group of more than two nodes. This structure has a significant impact on various dynamical processes taking place on networks. In order to capture higher-order interactions, some models were proposed, for many of which simplexes are elementary building blocks.

In this paper, we presented an extensive analytical study of resistance distances for a family of iteratively generating networks. This network family consists of simplexes presenting higher-order interactions among nodes, characterized by a tunable parameter $q$, which determines many structural properties, such as power-law degree distribution, clustering coefficient, and spectral dimension. We obtained analytical or exact expressions for two-node resistance distance, multiplicative degree-Kirchhoff index, additive degree-Kirchhoff index, Kirchhoff index, and average resistance distance. Our explicit results show that the average resistance distance converges to a $q$-dependent constant, implying that simplicial structure has strong influence on the structural robustness. Since diverse dynamical processes (for example, noisy consensus~\cite{YiZhShCh17, QiZhYiLi19}) are determined by average resistance distance, we confidently conclude that higher-order organization in network structure greatly affects the network dynamics.

Note that the network $\mathcal{G}_q(t)$ under study can  be considered as a pure $(q+1)$-dimensional simplicial complex, by looking upon each $(q+1)$-simplex and all its faces as the constituent  simplices.  Thus, the properties and high-order interaction of  $\mathcal{G}_q(t)$ can be characterized by a corresponding weighted graph $\mathcal{G}'_q(t)$~\cite{CoBi16,CoBi17,CaBaCeFa20}, which is obtained from $\mathcal{G}_q(t)$ by assigning an approximate weight to each edge in $\mathcal{G}_q(t)$. In future, we will study various dynamical processes running in $\mathcal{G}_q(t)$, in order to explore the effects of the high-order interactions on dynamics~\cite{BaCeLaLqLuPaYoPe20}.


\section*{Acknowledgements}

This work was supported by the National
Natural Science Foundation of China (Nos. 61872093 and U20B2051), the National
Key R \& D Program of China (No. 2018YFB1305104), Shanghai Municipal Science and Technology
Major Project  (Nos.  2018SHZDZX01 and 2021SHZDZX03),  ZJ Lab, and Shanghai Center for Brain Science and Brain-Inspired Technology. Mingzhe Zhu was also supported by Fudan's Undergraduate Research
Opportunities Program (FDUROP) under Grant No. 20001.

\section*{Data Availability Statement}

No new data were generated or analysed in support of this research.

\appendix

\section{Proof of Lemma~\ref{PFProA}}\label{AppA}

\begin{proof}
In order to prove $A^{\alpha,\beta}_{t+1}A^{\beta,\alpha}_{t+1}=q(D_t+A_t)$, it suffices to show that the corresponding entries of the two matrices $A^{\alpha,\beta}_{t+1}A^{\beta,\alpha}_{t+1}$ and $q(D_t+A_t)$ are equivalent to each other. For simplicity, let $Y_t=q(D_t+A_t)$, whose entries  are: $Y_t(i,j)=qd_i^{(t)}$, if $i=j$ and $Y_t(i,j)=qA_t(i,j)$ otherwise. Let $Z_t=A^{\alpha,\beta}_{t+1}A^{\beta,\alpha}_{t+1}$. Below it is verified that the elements $Z_t(i,j)$ of $Z_t$ equal  those associated with $Y_t$.

Notice that matrix $A^{\beta,\alpha}_{t+1}$ can be partitioned into $N_t$ column vectors $x_i=(x_{i,N_t+1},x_{i,N_t+2},\dots,x_{i,N_{t+1}})^\top$ $(i=1,2,\dots, N_t)$ as $A^{\beta,\alpha}_{t+1}=(x_1,x_2,\dots,x_{N_t})$, where each $x_i$ represents the adjacency relation between the old node $i\in\alpha$ and all new nodes in $\beta$. Since $A^{\alpha,\beta}_{t+1}=(A^{\beta,\alpha}_{t+1})^\top$, one obtains $A^{\alpha,\beta}_{t+1}=(x_1,x_2,\dots,x_{N_t})^\top$. Hence,
	\begin{equation}
		\begin{aligned}
		A^{\alpha,\beta}_{t+1}A^{\beta,\alpha}_{t+1}=&(x_1,x_2,\dots,x_{N_t})^\top(x_1,x_2,\dots,x_{N_t})\\
		=&(x_i^\top x_j)_{N_t\times N_t},
		\end{aligned}
	\end{equation}
 making use of which the entries $Z_t(i,j)$ of $A^{\alpha,\beta}_{t+1}A^{\beta,\alpha}_{t+1}$ can be determined by distinguishing two cases: $i = j$ and $i\neq j$.

 For the first case of $i=j$, the diagonal entry  $Z_t(i,i)$ of $Z_t$ is $Z_t(i,i)=x_i^\top x_i$,
equal to  the number of neighbors for  node $i$, which belongs  to set $\beta$. Thus, $Z_t(i,i)=d^{(t+1)}_i-d^{(t)}_i=qd^{(t)}_i=Y_t(i,i)$.

For the second case of $i\neq j$, the non-diagonal entry  $Z_t(i,j)$ of matrix $Z_t$ is
	\begin{equation}
	\begin{aligned}
	Z_t(i,j)=&x_i^\top x_j=\sum_{k\in \beta}(x_{i,k}\cdot x_{j,k})\\
	=&\sum_{k\in \beta}\left(A_{t+1}(i,k)A_{t+1}(j,k)\right)\\
	=&\sum_{\substack{A_{t+1}(i,k)=1\\A_{t+1}(j,k)=1}}A_t(i,j)\\
	=&qA_t(i,j)=Y_t(i,j),
	\end{aligned}
	\end{equation}
where it has used the construction that each edge in
	$\mathcal{G}_q(t)$ creates $q$ new nodes at iteration $t + 1$.
\end{proof}

\section{Proof of Lemma~\ref{beauty}}\label{AppB}

\begin{proof}
Matrix  $(q+1)I-A^{\beta,\beta}_{t+1}$ is in fact  a
diagonal block matrix of order $W_{t+1}/ q$, given by
$(q+1)I-A^{\beta,\beta}_{t+1}={\rm {diag}}(C_q,C_q\ldots,C_q)$
where $C_q$ is a $q\times q$ matrix, with  the diagonal and non-diagonal entries being $q+1$ and $-1$, respectively.

Define $Q$ to be the inverse of matrix $(q+1)I-A^{\beta,\beta}_{t+1}$. Then, it can be written as  a
diagonal block matrix of order $W_{t+1}/ q$ as $Q={\rm {diag}}(F_q,F_q\ldots,F_q)$, where $F_q$ is a $q\times q$ matrix, with  the diagonal and non-diagonal entries being $3/(2q+4)$ and $1/(2q+4)$, respectively.

Let $T=A^{\alpha,\beta}_{t+1}Q$. Then, for  $i \in \mathcal{V}_{t}$ and  $j \in\mathcal{W}_{t+1}$, its $(i,j)$th entry $T(i,j)$ is
	\begin{equation}
	\begin{split}
	T(i,j)&=\sum_{k \in\mathcal{W}_{t+1}} A^{\alpha,\beta}_{t+1}(i,k)Q(k,j)\\
	&=\sum_{(i,k)\in\mathcal{E}_{t+1}}\! Q(k,j).\\
	\end{split}
	\end{equation}
Note that the entry $Q(k,j)$ has only three possible values: 0, $\frac{3}{2(q+2)}$, and  $\frac{1}{2(q+2)}$. For  $k=j$, $Q(k,j)=\frac{3}{2(q+2)}$. For  $k\neq j$, one can distinguish two cases: if  $k\sim j$, $Q(k,j)=\frac{1}{2(q+2)}$; $Q(k,j)=0$ otherwise.  For $A^{\alpha,\beta}_{t+1}(i,j)$, it is $1$ if $i \sim j$, and it is $0$ if $i \nsim j$.

In order to determine $T(i,j)$, it suffices to distinguish two cases:  $i \sim j$ and  $i \nsim j$.  For the first case of $i \sim j$,
	\begin{equation}
	\begin{split}
	T(i,j)&=\sum_{(i,k)\in\mathcal{E}_{t+1}}\! Q(k,j)\\
	&=\frac{3}{2(q+2)}+(q-1)\frac{1}{2(q+2)}\\
	&=\frac{1}{2}= \frac{1}{2}A^{\alpha,\beta}_{t+1}(i,j).
	\end{split}
	\end{equation}
For the second case of node $i \nsim j$, one has $j \nsim k$ and $Q(k,j)=0$. Thus,
\begin{equation}
	T(i,j)=\sum_{(i,k)\in\mathcal{E}_{t+1}}\!Q(k,j)=0= \frac{1}{2}A^{\alpha,\beta}_{t+1}(i,j).
\end{equation}
Then,
	\begin{equation}
	T=A^{\alpha,\beta}_{t+1}\left((q+1)I-A^{\beta,\beta}_{t+1}\right)^{-1}=\frac{1}{2}A^{\alpha,\beta}_{t+1}.
	\end{equation}
This competes the proof.
\end{proof}

\section{Proof of Theorem~\ref{lemma14}}\label{AppC}

	\begin{proof}
		First,  establish the recursion relation governing $R^{\ast}_{\mathcal{V}_{t+1},\mathcal{V}_{t+1}}(t+1)$
		and $R^{\ast}_{\mathcal{V}_{t},\mathcal{V}_{t}}(t)$. 
		By definition,
		\begin{equation}\label{eq141}
			\begin{split}
				R^{\ast}_{\mathcal{V}_{t+1},\mathcal{V}_{t+1}}(t+1)=&R^{\ast}_{\alpha,\alpha}(t+1)+R^{\ast}_{\alpha,\beta}(t+1)+\\
				&R^{\ast}_{\beta,\alpha}(t+1)+R^{\ast}_{\beta,\beta}(t+1),
			\end{split}
		\end{equation}
		in which the four terms on the right-hand side are 	 evaluated as follows.
		
		For the first term, by Lemma~\ref{lemma7},
		\begin{equation}\label{eq142}
			\begin{split}
				R^{\ast}_{\alpha,\alpha}(t+1)&=\sum_{\substack{i\in\mathcal{V}_t\\j\in\mathcal{V}_t}}d^{(t+1)}_id^{(t+1)}_j\Omega^{(t+1)}_{ij}\\
				&=(q+1)^2\!\times\!\frac{2}{q+2}\sum_{\substack{i\in\mathcal{V}_t\\j\in\mathcal{V}_t}}\!d^{(t)}_id^{(t)}_j\Omega^{(t)}_{ij}\\
				&=\frac{2(q+1)^2}{q+2}R^{\ast}_{\mathcal{V}_{t},\mathcal{V}_{t}}(t).
			\end{split}
		\end{equation}
		
		The second and third terms are equivalent  to each other. 	By  using Lemma~\ref{lemma10} and the fact that $d^{(t+1)}_i=q+1$ for any node $i\in\mathcal{W}_{t+1}$, one has
		\begin{small}
			\begin{align}\label{eq143}
				&R^{\ast}_{\alpha,\beta}(t+1)=R^{\ast}_{\beta,\alpha}(t+1)\notag\\
				=&\sum_{\substack{i\in\mathcal{W}_{t+1}\\j\in\mathcal{V}_t}}d^{(t+1)}_id^{(t+1)}_j\Omega^{(t+1)}_{ij}\notag\\
				=&\!\sum_{\substack{i\in\mathcal{W}_{t+1}\\j\in\mathcal{V}_t}}\!\!\!(q\!+\!1)\left((q\!+\!1)d^{(t)}_j\right)\frac{1}{2}\left(\frac{3}{q\!+\!2}\!-\!\frac{1}{2}\Omega^{(t+1)}_{\Delta_{i}}\!\!+\!\Omega^{(t+1)}_{\Delta_{i},j}\right)\notag\\
				=&\frac{q+1}{2}\sum_{i\in\mathcal{W}_{t+1}}\left(\frac{3}{q+2}-\frac{1}{2}\Omega^{(t+1)}_{\Delta_{i}}\right)\left(\sum_{j\in\mathcal{V}_{t}}(q+1)d^{(t)}_j\right)\notag\\
				&+\frac{(q+1)^2}{2}\sum_{\substack{i\in\mathcal{W}_{t+1}\\j\in\mathcal{V}_t}}d^{(t)}_j\Omega^{(t+1)}_{\Delta_{i},j}.
			\end{align}
		\end{small}
		Considering  Lemmas~\ref{lemma12} and~\ref{lemma13},~\eqref{eq143} can be written as
		\begin{small}
			\begin{align}\label{eq144}
				R^{\ast}_{\alpha,\beta}(t+1)=&(q+1)^2M_t\left(\frac{3}{q+2}W_{t+1}-\frac{q(N_t-1)}{q+2}\right)\notag\\
				&+\frac{q(q+1)^2}{2}\sum_{i,j\in\mathcal{V}_{t}}d^{(t)}_id^{(t)}_j\Omega^{(t+1)}_{ij}\notag\\
				=&(q\!+\!1)^2M_t\left(\frac{3}{q+2}W_{t+1}\!-\!\frac{q(N_t\!-\!1)}{q\!+\!2}\right)\notag\\
				&+\frac{q(q\!+\!1)^2}{q\!+\!2}R^{\ast}_{\mathcal{V}_{t},\mathcal{V}_{t}}(t).
			\end{align}
		\end{small}
Next, continue to determine the last term. For convenience,  use $i\sim j$ ($i\nsim j$) to denote that nodes $i$ and $j$ are  adjacent (non-adjacent). Then,
		\begin{align}\label{eq:Rast1}
			R^{\ast}_{\beta,\beta}(t+1)\notag=&\sum_{\substack{i,j\in\mathcal{W}_{t+1}\\i\neq j\\i\sim j}}d^{(t+1)}_id^{(t+1)}_j\Omega^{(t+1)}_{ij}\notag\\
			&+\sum_{\substack{i,j\in\mathcal{W}_{t+1}\\i\nsim j}}d^{(t+1)}_id^{(t+1)}_j\Omega^{(t+1)}_{ij}
		\end{align}
		Considering Lemma~\ref{lemma8}, the first term on the right-hand side of~\eqref{eq:Rast1} can be further written as
		\begin{align}\label{eq:term1}
			&\sum_{\substack{i,j\in\mathcal{W}_{t+1}\\i\neq j\\i\sim j}}d^{(t+1)}_id^{(t+1)}_j\Omega^{(t+1)}_{ij}\notag\\
			=&(q+1)^2\sum_{\substack{i,j\in\mathcal{W}_{t+1}\\i\neq j\\i\sim j}}\frac{2}{q+2}\notag\\
			=&\frac{2(q+1)^2}{q+2}W_{t+1}(q-1).
		\end{align}
		By using Lemma~\ref{lemma11}, the second term on the right-hand side of~\eqref{eq:Rast1} can be further written as
		\begin{small}
			\begin{align}\label{eq:ijnsim}
				&\sum_{\substack{i,j\in\mathcal{W}_{t+1}\\i\nsim j}}\!\!d^{(t+1)}_id^{(t+1)}_j\Omega^{(t+1)}_{ij}\notag\\
				=&(q\!+\!1)^2\!\!\!\sum_{\substack{i,j\in\mathcal{W}_{t\!+\!1}\\i\nsim j}}\!\!\left(\frac{3}{q\!+\!2}\!-\!\frac{1}{4}\left(\Omega^{(t+1)}_{\Delta_{i}}\!+\!\Omega^{(t+1)}_{\Delta_{j}}\right)\!+\!\frac{1}{4}\Omega^{(t+1)}_{\Delta_{i},\Delta_{j}}\right)\notag\\
				=&\frac{3(q\!+\!1)^2}{q\!+\!2}W_{t+1}(W_{t+1}\!\!-\!\!q)\!-\!\frac{(q\!+\!1)^2}{2}(W_{t+1}\!-\!1)\!\!\!\!\sum_{i\in\mathcal{W}_{t\!+\!1}}\!\!\!\Omega^{(t+1)}_{\Delta_i}\notag\\
				&+\frac{(q\!+\!1)^2}{4}\left(\sum_{i,j\in\mathcal{W}_{t\!+\!1}}\!\!\Omega^{(t+1)}_{\Delta_i,\Delta_j}\!-\!\sum_{i\in\mathcal{W}_{t\!+\!1}}\!\!\Omega^{(t+1)}_{\Delta_i,\Delta_i}\right).
			\end{align}
		\end{small}	
		From Lemma~\ref{lemma13} and Lemma~\ref{lemma7}, one has
		\begin{align}\label{eq:dxdy}
			\sum_{i,j\in\mathcal{W}_{t\!+\!1}}\Omega^{(t+1)}_{\Delta_i,\Delta_j}=&\sum_{j\in\mathcal{W}_{t+1}}\sum_{i\in\mathcal{W}_{t+1}}\Omega^{(t+1)}_{\Delta_i,\Delta_j}\notag\\
			=&\sum_{j\in\mathcal{W}_{t+1}}\sum_{x\in\mathcal{V}_{t}}qd_x^{(t)}\Omega_{x,\Delta_j}^{(t+1)}\notag\\
			=&\sum_{x\in\mathcal{V}_{t}}qd_x^{(t)}\sum_{j\in\mathcal{W}_{t+1}}\Omega_{x,\Delta_j}^{(t+1)}\notag\\
			=&q^2\sum_{x,y\in\mathcal{V}_{t}}d^{(t)}_xd^{(t)}_y\Omega^{(t+1)}_{xy}\notag\\
			=&\frac{2q^2}{q+2}\sum_{x,y\in\mathcal{V}_{t}}d^{(t)}_xd^{(t)}_y\Omega^{(t)}_{xy}\notag\\
			=&\frac{2q^2}{q+2}R^{\ast}_{\mathcal{V}_t,\mathcal{V}_t}(t).
		\end{align}
		Using $\sum_{i\in\mathcal{W}_{t+1}}\Omega^{(t+1)}_{\Delta_{i},\Delta_{i}}=2\sum_{i\in\mathcal{W}_{t+1}}\Omega^{(t+1)}_{\Delta_{i}}$, Lemma~\ref{lemma12} and~\eqref{eq:dxdy},~\eqref{eq:ijnsim} is reduced to
		\begin{align}\label{eq:term2}
			&\sum_{\substack{i,j\in\mathcal{W}_{t+1}\\i\nsim j}}d^{(t+1)}_id^{(t+1)}_j\Omega^{(t+1)}_{ij}\notag\\
			=&\frac{(q+1)^2}{q+2}W_{t+1}\big(3W_{t+1}-3q-q(N_t-1)\big)\notag\\
			&+\frac{q^2(q+1)^2}{2(q+2)}R^{\ast}_{\mathcal{V}_t,\mathcal{V}_t}(t)
		\end{align}
		With~\eqref{eq:term1} and~\eqref{eq:term2}, $R^{\ast}_{\beta,\beta}(t+1)$ can be further evaluated by
		\begin{align}\label{eq145}
			&R^{\ast}_{\beta,\beta}(t+1)\notag\\
			=&\frac{(q+1)^2}{q+2}W_{t+1}\big(3W_{t+1}-q-2-q(N_t-1)\big)\notag\\
			&+\frac{q^2(q+1)^2}{2(q+2)}R^{\ast}_{\mathcal{V}_t,\mathcal{V}_t}(t).
		\end{align}
			
Substituting~\eqref{eq142},~\eqref{eq144} and~\eqref{eq145} into~\eqref{eq141} leads to
		\begin{equation*}
			\begin{split}
				&R^{\ast}_{\mathcal{V}_{t+1},\mathcal{V}_{t+1}}(t+1)\\
				=&\frac{(q+2)(q+1)^2}{2}R^{\ast}_{\mathcal{V}_t,\mathcal{V}_t}(t)\\
				&-\frac{q(q+2)^2(q+1)^3}{q+3}\left(\frac{(q+1)(q+2)}{2}\right)^t\\
				&+\frac{q(3q+7)(q+2)^2(q+1)^4}{4(q+3)}\left(\frac{(q+1)(q+2)}{2}\right)^{2t}.
			\end{split}
		\end{equation*}
		which, under the condition $R^{\ast}_{\mathcal{V}_0,\mathcal{V}_0}(0)=2(q+1)^3$, can be solved to  obtain the desired result.
	\end{proof}

\section{Proof of Theorem~\ref{lemma15}}\label{AppD}

\begin{proof}
First, establish the recursive relation between $R^+_{\mathcal{V}_t,\mathcal{V}_t}(t)$ and $R^+_{\mathcal{V}_{t+1},\mathcal{V}_{t+1}}(t+1)$. By definition, one has
	\begin{align}\label{eq151}
	&R^{+}_{\mathcal{V}_{t+1},\mathcal{V}_{t+1}}(t+1)\notag\\
=&R^{+}_{\alpha,\alpha}(t\!+\!1)\!+\!R^{+}_{\alpha,\beta}(t\!+\!1)\!+\!R^{+}_{\beta,\alpha}(t\!+\!1)\!+\!R^{+}_{\beta,\beta}(t\!+\!1)\notag\\
	=&R^{+}_{\alpha,\alpha}(t+1)+2R^{+}_{\alpha,\beta}(t+1)+R^{+}_{\beta,\beta}(t+1).
	\end{align}
For the three terms on the right-hand side of~\eqref{eq151}, it is straightforward to verify that
	\begin{equation*}
	\begin{split}
	&R^{+}_{\alpha,\alpha}(t+1)=\frac{2(q+1)}{q+2}R^{+}_{\alpha,\alpha}(t),\\
	&2R^{+}_{\alpha,\beta}(t+1)=2\sum_{\substack{i\in\mathcal{V}_{t},\\j\in\mathcal{W}_{t+1}}}\left((q+1)+d^{(t+1)}_i\right)\Omega^{(t+1)}_{ij}\\
	&=2(q+1)R_{\alpha,\beta}(t+1)+\frac{2}{q+1}R^{\ast}_{\alpha,\beta}(t+1),\\
	&R^{+}_{\beta,\beta}(t+1)=\frac{2}{q+1}R^{\ast}_{\beta,\beta}(t+1).
	\end{split}
	\end{equation*}
Then, the quantity $R^{+}_{\mathcal{V}_{t+1},\mathcal{V}_{t+1}}(t+1)$ can be computed by
	\begin{equation}\label{eq:12x}
	\begin{split}
	&R^{+}_{\mathcal{V}_{t+1},\mathcal{V}_{t+1}}(t+1)=\\
	&\frac{2(q+1)}{q+2}R^{+}_{\alpha,\alpha}(t)+2(q+1)R_{\alpha,\beta}(t+1)\\
	&+\frac{2}{q+1}R^{\ast}_{\alpha,\beta}(t+1)+\frac{2}{q+1}R^{\ast}_{\beta,\beta}(t+1).
	\end{split}
	\end{equation}
The quantities on the right-hand side of~\eqref{eq:12x} have been previously determined, with the exception of  $R_{\alpha,\beta}(t+1)$.  Hence, in order to evaluate $R^{+}_{\mathcal{V}_{t+1},\mathcal{V}_{t+1}}(t+1)$, one only needs to  evaluate $R_{\alpha,\beta}(t+1)$. Note that
	\begin{equation}\label{eq152}
	\begin{split}
	\sum_{i,j\in\mathcal{V}_{t}}d^{(t)}_i\Omega^{(t)}_{ij}&=\frac{1}{2}\sum_{i,j\in\mathcal{V}_{t}}(d^{(t)}_i+d^{(t)}_j)\Omega^{t}_{ij}\\
	&=\frac{1}{2}R^{+}_{\mathcal{V}_{t},\mathcal{V}_{t}}(t).
	\end{split}
	\end{equation}
Making use of  Lemmas~\ref{lemma10},~\ref{lemma12} and~\ref{lemma13} and~\eqref{eq152}, one obtains
	\begin{equation}\label{eq153}
	\begin{split}
	&R_{\alpha,\beta}(t+1)=R_{\beta,\alpha}(t+1)
	=\sum_{\substack{i\in\mathcal{W}_{t+1}\\j\in\mathcal{V}_t}}\Omega^{(t+1)}_{ij}\\
	=&\sum_{\substack{i\in\mathcal{W}_{t+1}\\j\in\mathcal{V}_t}}\frac{1}{2}\left(\frac{3}{q+2}-\frac{1}{2}\Omega^{(t+1)}_{\Delta_{i}}+\Omega^{(t+1)}_{\Delta_{i},j}\right)\\
	=&\sum_{\substack{i\in\mathcal{W}_{t+1}\\j\in\mathcal{V}_t}}\!\!\frac{1}{2}\left(\frac{3}{q\!+\!2}\!-\!\frac{1}{2}\Omega^{(t+1)}_{\Delta_{i}}\right)\!+\!\frac{q}{q\!+\!2}\sum_{i,j\in\mathcal{V}_{t}}d^{(t)}_i\Omega^{(t)}_{ij}\\
	=&\frac{1}{2}N_t\!\left(\frac{3}{q\!+\!2}W_{t+1}\!-\!\frac{q(N_t\!-\!1)}{q\!+\!2}\right)\!\!+\!\frac{q}{2(q\!+\!2)}R^{+}_{\mathcal{V}_t,\mathcal{V}_t}(t).
	\end{split}
	\end{equation}
Inserting the above-obtained results into~\eqref{eq151} gives
	\begin{small}
		\begin{equation}\label{eqx}
		\begin{aligned}
		&R^{+}_{\mathcal{V}_{t+1},\mathcal{V}_{t+1}}(t+1)\\
		=&(q+1)R^{+}_{\mathcal{V}_t,\mathcal{V}_t}(t)+q(q+1)R^{\ast}_{\mathcal{V}_t,\mathcal{V}_t}(t)\\
		&+(q+1)(N_t+2M_t)\left(\frac{3}{q+2}W_{t+1}-\frac{q(N_t-1)}{q+2}\right)\\
		&+\frac{2(q+1)}{q+2}W_{t+1}\big(3W_{t+1}-q-2-q(N_t-1)\big).
		\end{aligned}
		\end{equation}
	\end{small}
	Considering  $R^{+}_{\mathcal{V}_{0},\mathcal{V}_{0}}(0)=4(q+1)^2$ and Theorem~\ref{lemma14}, ~\eqref{eqx} can be solved to obtain~\eqref{eqxxx}.
\end{proof}

\section{Proof of Theorem~\ref{lemma16}}\label{AppE}

\begin{proof}
By definition, one has
\begin{align}\label{eq161}
&R_{\mathcal{V}_{t+1},\mathcal{V}_{t+1}}(t+1)\notag\\
=&R_{\alpha,\alpha}(t+1)+2R_{\alpha,\beta}(t+1)+R_{\beta,\beta}(t+1)\\
=&\frac{2}{q\!+\!2}R_{\mathcal{V}_{t},\mathcal{V}_{t}}(t)\!+\!2R_{\alpha,\beta}(t\!+\!1)\!+\!\frac{1}{(q\!+\!1)^2}R^{\ast}_{\beta,\beta}(t\!+\!1)\notag.
\end{align}
Inserting~\eqref{eq145} and~\eqref{eq153} into~\eqref{eq161}, one obtains the recursive relation for $R_{\mathcal{V}_{t},\mathcal{V}_{t}}(t)$ as
\begin{small}
	\begin{align}\label{eq162}
	&R_{\mathcal{V}_{t+1},\mathcal{V}_{t+1}}(t+1)\notag\\
	=&\frac{2}{q+2}R_{\mathcal{V}_{t},\mathcal{V}_{t}}(t)+\frac{q}{q+2}R^{+}_{\mathcal{V}_{t},\mathcal{V}_{t}}(t)+\frac{q^2}{2(q+2)}R^{\ast}_{\mathcal{V}_{t},\mathcal{V}_{t}}(t)\notag\\
	&+N_t\left(\frac{3}{q+2}W_{t+1}-\frac{q(N_t-1)}{q+2}\right)\notag\\
	&+\frac{1}{q+2}W_{t+1}\Big(3W_{t+1}-q-2-q(N_t-1)\Big).
	\end{align}
\end{small}
Using Theorems~\ref{lemma14} and~\ref{lemma15} and $R_{\mathcal{V}_{0},\mathcal{V}_{0}}(0)=2(q+1)$,~\eqref{eq162} can be solved recursively to obtain~\eqref{eq160}.
\end{proof}

\bibliographystyle{compj}
\bibliography{Triangulation,Hitting}

\providecommand{\noopsort}[1]{}\providecommand{\singleletter}[1]{#1}%
\begin{thebibliography}{99}

\bibitem{Ne03}
Newman, M.~E. (2003) The structure and function of complex networks.
\newblock {\em SIAM Rev.}, {\bf  45}, 167--256.

\bibitem{Ba16}
Barab{\'a}si, A.-L. et al. (2016) {\em Network Science}. Cambridge university
  press.

\bibitem{BeAbScJaKl18}
Benson, A.~R., Abebe, R., Schaub, M.~T., Jadbabaie, A., and Kleinberg, J.
  (2018) Simplicial closure and higher-order link prediction.
\newblock {\em Proc. Natl. Acad. Sci.}, {\bf  115}, E11221--E11230.

\bibitem{SaCaDaLa18}
Salnikov, V., Cassese, D., and Lambiotte, R. (2018) Simplicial complexes and
  complex systems.
\newblock {\em Eur. J. Phys.}, {\bf  40}, 014001.

\bibitem{BeGlLe16}
Benson, A.~R., Gleich, D.~F., and Leskovec, J. (2016) Higher-order organization
  of complex networks.
\newblock {\em Science}, {\bf  353}, 163--166.

\bibitem{GrBaMiAl17}
Grilli, J., Barab{\'a}s, G., Michalska-Smith, M.~J., and Allesina, S. (2017)
  Higher-order interactions stabilize dynamics in competitive network models.
\newblock {\em Nature}, {\bf  548}, 210.

\bibitem{LaPeBaLa19}
Lacopini, I., Petri, G., Barrat, A., and Latora, V. (2019) Simplicial models of
  social contagion.
\newblock {\em Nat. Commun.}, {\bf  10}, 2485.

\bibitem{PaPeVa17}
Patania, A., Petri, G., and Vaccarino, F. (2017) The shape of collaborations.
\newblock {\em EPJ Data Sci.}, {\bf  6}, 18.

\bibitem{GiPaCuIt15}
Giusti, C., Pastalkova, E., Curto, C., and Itskov, V. (2015) Clique topology
  reveals intrinsic geometric structure in neural correlations.
\newblock {\em Proc. Natl. Acad. Sci.}, {\bf  112}, 13455--13460.

\bibitem{ReNoScect17}
Reimann, M.~W., Nolte, M., Scolamiero, M., Turner, K., Perin, R., Chindemi, G.,
  D{\l}otko, P., Levi, R., Hess, K., and Markram, H. (2017) Cliques of neurons
  bound into cavities provide a missing link between structure and function.
\newblock {\em Front. Comput. Neurosci.}, {\bf  11}, 48.

\bibitem{WuOtBa03}
Wuchty, S., Oltvai, Z.~N., and Barab{\'a}si, A.-L. (2003) Evolutionary
  conservation of motif constituents in the yeast protein interaction network.
\newblock {\em Nature Genetics}, {\bf  35}, 176.

\bibitem{CoBi17}
Courtney, O.~T. and Bianconi, G. (2017) Weighted growing simplicial complexes.
\newblock {\em Phys. Rev. E}, {\bf  95}, 062301.

\bibitem{PeBa18}
Petri, G. and Barrat, A. (2018) Simplicial activity driven model.
\newblock {\em Phys. Rev. Lett.}, {\bf  121}, 228301.

\bibitem{GiGhBa16}
Giusti, C., Ghrist, R., and Bassett, D.~S. (2016) Two's company, three (or
  more) is a simplex.
\newblock {\em J. Comput. Neurosci.}, {\bf  41}, 1--14.

\bibitem{CoBi18}
Courtney, O.~T. and Bianconi, G. (2018) Dense power-law networks and simplicial
  complexes.
\newblock {\em Phys. Rev. E}, {\bf  97}, 052303.

\bibitem{daBiGiDoMe18}
da Silva, D.~C., Bianconi, G., da Costa, R.~A., Dorogovtsev, S.~N., and Mendes,
  J.~F. (2018) Complex network view of evolving manifolds.
\newblock {\em Phys. Rev. E}, {\bf  97}, 032316.

\bibitem{QiYiZh19}
Qi, Y., Yi, Y., and Zhang, Z. (2019) Topological and spectral properties of
  small-world hierarchical graphs.
\newblock {\em Comput. J.}, {\bf  62}, 769--784.

\bibitem{BiRa17}
Bianconi, G. and Rahmede, C. (2017) Emergent hyperbolic network geometry.
\newblock {\em Sci. Rep.}, {\bf  7}, 41974.

\bibitem{BiZi18}
Bianconi, G. and Ziff, R.~M. (2018) Topological percolation on hyperbolic
  simplicial complexes.
\newblock {\em Phys. Rev. E}, {\bf  98}, 052308.

\bibitem{SkAr19}
Skardal, P.~S. and Arenas, A. (2019) Abrupt desynchronization and extensive
  multistability in globally coupled oscillator simplexes.
\newblock {\em Phys. Rev. Lett.}, {\bf  122}, 248301.

\bibitem{MaGoAr20}
Matamalas, J.~T., G{\'o}mez, S., and Arenas, A. (2020) Abrupt phase transition
  of epidemic spreading in simplicial complexes.
\newblock {\em Phys. Rev. Research}, {\bf  2}, 012049.

\bibitem{WaYiXuZh19}
Wang, Y., Yi, Y., Xu, W., and Zhang, Z. (2021) Modeling higher-order
  interactions in complex networks by edge product of graphs.
\newblock {\em Comput. J. (in press)}, {\bf  **}, ***--***.

\bibitem{SpSr11}
Spielman, D.~A. and Srivastava, N. (2011) Graph sparsification by effective
  resistances.
\newblock {\em SIAM J. Comput.}, {\bf  40}, 1913--1926.

\bibitem{DoBu12}
Dorfler, F. and Bullo, F. (2013) Kron reduction of graphs with applications to
  electrical networks.
\newblock {\em IEEE Trans. Circuits Syst. I: Regular Papers}, {\bf  60},
  150--163.

\bibitem{YoScLe15}
Young, G.~F., Scardovi, L., and Leonard, N.~E. (2015) A new notion of effective
  resistance for directed graphs---{P}art \textsc{I}: {D}efinition and
  properties.
\newblock {\em IEEE Trans. Autom. Control}, {\bf  61}, 1727--1736.

\bibitem{ThYaNa18}
Thulasiraman, K., Yadav, M., and Naik, K. (2019) Network science meets circuit
  theory: {R}esistance distance, {K}irchhoff index, and {F}oster's theorems
  with generalizations and unification.
\newblock {\em IEEE Trans. Circuits Syst. I: Regular Papers}, {\bf  66},
  1090--1103.

\bibitem{DoSibu18}
D{\"o}rfler, F., Simpson-Porco, J.~W., and Bullo, F. (2018) Electrical networks
  and algebraic graph theory: {M}odels, properties, and applications.
\newblock {\em Proc. IEEE}, {\bf  106}, 977--1005.

\bibitem{ShZh19}
Sheng, Y. and Zhang, Z. (2019) Low mean hitting time for random walks on
  heterogeneous networks.
\newblock {\em IEEE Trans. Inf. Theory}, {\bf  65}, 6898--6910.

\bibitem{SoHiLi19}
Song, Y., Hill, D.~J., and Liu, T. (2019) On extension of effective resistance
  with application to graph laplacian definiteness and power network stability.
\newblock {\em IEEE Trans. Circuits Syst. I: Regular Papers}, {\bf  66},
  4415--4428.

\bibitem{CoBi16}
Courtney, O.~T. and Bianconi, G. (2016) Generalized network structures: The
  configuration model and the canonical ensemble of simplicial complexes.
\newblock {\em Phys. Rev. E}, {\bf  93}, 062311.

\bibitem{CaBaCeFa20}
Carletti, T., Battiston, F., Cencetti, G., and Fanelli, D. (2020) Random walks
  on hypergraphs.
\newblock {\em Phys. Rev. E}, {\bf  101}, 022308.

\bibitem{Ti94}
Tian, Y. (1994) Reverse order laws for the generalized inverses of multiple
  matrix products.
\newblock {\em Linear Algebra Appl.}, {\bf  211}, 85--100.

\bibitem{Li19}
Liu, Q. (2019) The resistance distance and {K}irchhoff index on quadrilateral
  graph and pentagonal graph.
\newblock {\em IEEE Access}, {\bf  7}, 36617--36622.

\bibitem{DoSn84}
Doyle, P.~G. and Snell, J.~L. (1984) {\em Random Walks and Electric Networks}.
  Mathematical Association of America.

\bibitem{KlRa93}
Klein, D.~J. and Randi{\'c}, M. (1993) Resistance distance.
\newblock {\em J. Math. Chem.}, {\bf  12}, 81--95.

\bibitem{Ba99}
Bapat, R. (1999) Resistance distance in graphs.
\newblock {\em Math. Student}, {\bf  68}, 87--98.

\bibitem{Kl02}
Klein, D.~J. (2002) Resistance-distance sum rules.
\newblock {\em Croat. Chem. Acta}, {\bf  75}, 633--649.

\bibitem{FoRo1949}
Foster, R.~M. (1949) The average impedance of an electrical network.
\newblock {\em Contributions to Applied Mechanics (Reissner Anniversary
  Volume)} , {\bf ?}, 333--340.

\bibitem{Fo61}
Foster, R. (1961) An extension of a network theorem.
\newblock {\em IRE Trans. Circuit Theory}, {\bf  8}, 75--76.

\bibitem{Ch10}
Chen, H. (2010) Random walks and the effective resistance sum rules.
\newblock {\em Discrete Appl. Math.}, {\bf  158}, 1691--1700.

\bibitem{GhBoSa08}
Ghosh, A., Boyd, S., and Saberi, A. (2008) Minimizing effective resistance of a
  graph.
\newblock {\em SIAM Rev.}, {\bf  50}, 37--66.

\bibitem{TiLe10}
Tizghadam, A. and Leon-Garcia, A. (2010) Autonomic traffic engineering for
  network robustness.
\newblock {\em IEEE J. Sel. Areas Commun.}, {\bf  28}.

\bibitem{LiZh18}
Li, H. and Zhang, Z. (2018) Kirchhoff index as a measure of edge centrality in
  weighted networks: {N}early linear time algorithms.
\newblock {\em Proceedings of the 29th Annual ACM-SIAM Symposium on Discrete
  Algorithms},  San Francisco, California USA,  23-25 January,  pp. 2377--2396.
  SIAM.

\bibitem{PaBa14}
Patterson, S. and Bamieh, B. (2014) Consensus and coherence in fractal
  networks.
\newblock {\em IEEE Trans. Control Netw. Syst.}, {\bf  1}, 338--348.

\bibitem{QiZhYiLi19}
Qi, Y., Zhang, Z., Yi, Y., and Li, H. (2019) Consensus in self-similar
  hierarchical graphs and {S}ierpi{\'n}ski graphs: {C}onvergence speed, delay
  robustness, and coherence.
\newblock {\em IEEE Trans. Cybern.}, {\bf  49}, 592--603.

\bibitem{YiZhPa20}
Yi, Y., Zhang, Z., and Patterson, S. (2020) Scale-free loopy structure is
  resistant to noise in consensus dynamics in complex networks.
\newblock {\em IEEE Trans. Cybern.}, {\bf  50}, 190--200.

\bibitem{ShCaHu18}
Shi, X., Cao, J., and Huang, W. (2018) Distributed parametric consensus
  optimization with an application to model predictive consensus problem.
\newblock {\em IEEE Trans. Cybern.}, {\bf  48}, 2024--2035.

\bibitem{ZhXuYiZh22}
Zhang, Z., Xu, W., Yi, Y., and Zhang, Z. (2022) Fast approximation of coherence
  for second-order noisy consensus networks.
\newblock {\em IEEE Trans. Cybern.}, {\bf  52}, 677--686.

\bibitem{YiYaZhZhPa22}
Yi, Y., Yang, B., Zhang, Z., Zhang, Z., and Patterson, S. (2022) Biharmonic
  distance-based performance metric for second-order noisy consensus networks.
\newblock {\em IEEE Trans. Inf. Theory}, {\bf  68}, 1220--1236.

\bibitem{ChZh07}
Chen, H. and Zhang, F. (2007) {Resistance distance and the normalized Laplacian
  spectrum}.
\newblock {\em Discrete Appl. Math.}, {\bf  155}, 654--661.

\bibitem{GuFeYu12}
Gutman, I., Feng, L., and Yu, G. (2012) Degree resistance distance of unicyclic
  graphs.
\newblock {\em Trans. Combin.}, {\bf  1}, 27--40.

\bibitem{Hu14}
Hunter, J.~J. (2014) {The role of Kemeny's constant in properties of Markov
  chains}.
\newblock {\em Commun. Stat. --- Theor. Methods}, {\bf  43}, 1309--1321.

\bibitem{XuShZhKaZh20}
Xu, W., Sheng, Y., Zhang, Z., Kan, H., and Zhang, Z. (2020) Power-law graphs
  have minimal scaling of {K}emeny constant for random walks.
\newblock {\em Proceedings of The Web Conference},  pp. 46--56.

\bibitem{LeLo02}
Levene, M. and Loizou, G. (2002) Kemeny's constant and the random surfer.
\newblock {\em Am. Math. Mon.}, {\bf  109}, 741--745.

\bibitem{PaAgBu15}
Patel, R., Agharkar, P., and Bullo, F. (2015) {Robotic surveillance and Markov
  chains with minimal weighted Kemeny constant}.
\newblock {\em IEEE Trans. Autom. Control}, {\bf  60}, 3156--3167.

\bibitem{JaOl19}
Jadbabaie, A. and Olshevsky, A. (2019) Scaling laws for consensus protocols
  subject to noise.
\newblock {\em IEEE Trans. Autom. Control}, {\bf  64}, 1389--1402.

\bibitem{YiZhShCh17}
Yi, Y., Zhang, Z., Shan, L., and Chen, G. (2017) Robustness of first-and
  second-order consensus algorithms for a noisy scale-free small-world {K}och
  network.
\newblock {\em IEEE Trans. Control Syst. Technol.}, {\bf  25}, 342--350.

\bibitem{BaCeLaLqLuPaYoPe20}
Battiston, F., Cencetti, G., Lacopini, I., Latora, V., Lucas, M., Patania, A.,
  Young, J.-G., and Petri, G. (2020) Networks beyond pairwise interactions:
  {S}tructure and dynamics.
\newblock {\em Phys. Rep.}, {\bf  874}, 1--92.

\end{thebibliography}

\end{document}